\def\final{1}
\def\pqc{0}

\ifnum\pqc=1
\ifnum\final=1
\documentclass{llncs}
\else
\documentclass[runningheads]{llncs}
\fi
\else
\documentclass[11pt]{article}
\usepackage{amsthm}
\fi
\usepackage{mathpazo}
\usepackage{times}
\usepackage{subfigure}
\usepackage{amsfonts}
\usepackage{multicol}
\usepackage{graphicx}
\usepackage{amssymb}
\usepackage{amsmath}
\usepackage{latexsym}
\usepackage[usenames]{color}
\usepackage{paralist}
\usepackage{hyperref}
\usepackage{soul}
\usepackage{verbatim}
\usepackage{mdwlist}
\usepackage{algorithmic}
\usepackage{algorithm}

\ifnum\pqc=0

 \hypersetup{pdfpagemode=UseNone}

\fi

\ifnum\final=0
\newcommand{\fnote}[1]{\textbf{\color{red} Note}: #1}
\newcommand{\ftodo}[1]{\textbf{\color{blue} Next: #1}}
\else
\newcommand{\fnote}[1]{}
\newcommand{\ftodo}[1]{}
\fi

\def\natural{\mathbb{N}}

\def\({\left(}
\def\){\right)}

\newcommand{\mypar}[1]{\smallskip \noindent {\sc {#1}.}}

\def\calH{\mathcal{H}}
\def\calO{\mathcal{O}}
\def\calS{\mathcal{S}}
\def\veps{\varepsilon}

\def\red{\mathcal{R}}
\def\qred{\hat{\mathcal{R}}}
\def\qintp{\hat{\mathcal{I}}}
\def\trr{\mathcal{T}}
\def\btrr{T} 
\def\qtrr{\hat{\mathcal{T}}}
\newcommand{\qver}[1]{\hat{#1}}
\def\disr{\mathcal{D}}
\def\adv{\mathcal{A}}
\def\badv{\mathcal{B}}
\def\qadv{{\hat{\mathcal{A}}}}
\def\bqadv{{\hat{\mathcal{B}}}}
\def\chr{\mathcal{C}}

\def\extg{{G^{\sf ext}}}
\def\intg{G^{\sf int}}
\def\extgq{\hat G^{\sf ext}}
\def\intgq{\hat G^{\sf int}}
\def\preal{\mathbb{R}^+}
\def\negl{\text{\sf negl}}

\def\game{G}
\def\qgame{\hat G}
\def\succ{\sf succ}
\def\fail{\sf fail}
\def\asucc{\alpha_{\sf succ}}
\def\atime{\alpha_{\sf time}}

\def\qpoly{\mathcal{Q}}
\def\clsa{\mathfrak{A}}
\def\qclsa{{\hat{\mathfrak{A}}}}
\def\clsb{\mathfrak{B}}
\def\qclsb{\hat{\mathfrak{B}}}
\def\gcls{\mathfrak{C}}
\def\sredcls{( \extg(\clsb) , {\trr}, \intg(\clsa))} 

\def\respab{\beta\text{-}(\qclsa,\qclsb)}

\def\transred{$(\beta,\beta')$-$(\qclsa,\qclsb)$}

\def\realize{realizable}

\def\eucma{\text{\sc eu-cma}}
\def\sucma{\text{\sc su-cma}}
\def\owf{\text{\sc owf}}
\def\owfs{\text{\sc owfs}}
\def\ots{\text{\sc ots}}
\def\lots{\text{\sc l-ots}}
\def\wots{\text{\sc w-ots}}
\def\xmss{\text{\sc xmss}}
\def\uowhf{\text{\sc uowhfs}}
\def\ruowhf{\text{\sc r-h}}
\def\uhash{\mathcal{H} =\{h_s\}}
\def\tdp{\text{\sc tdp}}
\def\prg{\text{\sc prg}}
\def\prf{\text{\sc prf}}
\def\kow{\text{\sc kow}}
\def\spr{\text{\sc spr}}
\def\mtree{\text{\sc m-tree}}

\def\ginv{G^{\text{\sc inv}}} 
\def\ginvv{G^{\text{\sc inv'}}} 

\def\gcol{G^{\text{\sc col}}} 
\def\gcolv{G^{\text{\sc col'}}} 
\def\gcolvv{G^{\text{\sc col''}}} 
\def\gfor{G^{\text{\sc for}}}
\def\gforot{G^{\text{\sc ot-for}}}
\def\gforro{G^{\text{\sc ro-for}}}
\def\gforqro{G^{\text{\sc qro-for}}}
\def\gspr{G^{\text{\sc spr}}}
\def\gprf{G^{\text{\sc prf}}}
\def\gtdp{G^{\text{\sc tdp}}}

\def\zkdistg{G^{ZK}_{V^*,\calS}}
\def\qzkdistg{\hat G^{ZK}_{\hat V^*,\hat\calS}}

\def\sign{\text{\sf Sign}}
\def\kgen{\text{\sf KGen}}
\def\vrfy{\text{\sf Vrfy}}

\newcommand{\geqm}[2]{\lbrack #1,#2\rbrack}
\newcommand{\gval}[1]{\omega_{#1}}

\newcommand{\rtime}[1]{\text{\sf TIME}(#1)}
\newcommand{\eqmcls}[2]{E_{#1}(#2)}

\newcommand{\secpar}[1]{\veps_{#1}}

\newcommand{\tadv}[2]{{#1}(#2)}

\newcommand{\respred}[2]{$\beta$-$(#1,#2)$}

\newenvironment{prot}[2]{
\vspace{-2ex}
\begin{figure}[ht!]
\begin{center}
   \begin{tabular}{|ll|}
   \hline
     \hspace{.1ex}\begin{minipage}{.97\linewidth}\vspace{0.5ex}
       {\begin{center}
       {\bf #1} {#2} \end{center}}\vspace{-2ex}
       }{%
       \vspace{-5ex}
       \smallskip
     \end{minipage}& \\
     \hline
   \end{tabular}
\vspace*{1ex}
   \end{center}
\end{figure}
}

\newenvironment{ncprot}[2]{
\begin{figure}[ht!]
\begin{center}
   \begin{tabular}{|ll|}
   \hline
     \hspace{.1ex}\begin{minipage}{.97\linewidth}\vspace{0.5ex}
       {\begin{center}
       {\bf #1} {#2} \end{center}}\vspace{-2ex}
       }{%
       \vspace{-2ex}
       \smallskip
     \end{minipage}& \\
     \hline
   \end{tabular}
   \end{center}
   \vspace{-2ex}
\end{figure}
}
\ifnum\pqc=0
\newtheorem{theorem}{Theorem}[section]
\newtheorem{corollary}[theorem]{Corollary}

\newtheorem{proposition}[theorem]{Proposition}

\theoremstyle{definition}
\newtheorem{definition}[theorem]{Definition}
\fi

\begin{document}

\title{A Note on Quantum Security for Post-Quantum Cryptography}

\ifnum\pqc=1
\author{Fang Song}
\institute{Department of Combinatorics \& Optimization \\
and Institute for Quantum Computing\\
University of Waterloo\\
}

\else
\author{Fang Song \\
Department of Combinatorics \& Optimization \\
and Institute for Quantum Computing\\
University of Waterloo\\
}
\fi

\date{}
\maketitle


\ifnum\final=0
To-Do:
\begin{itemize}
	\item FDH in QRO: go over details again 
	\item Abstract: elusive?
	\item Intro: we give an alternative proof
	\item Ginv': another distribution? 
\end{itemize}
\fi

\begin{abstract}
Shor's quantum factoring algorithm and a few other efficient quantum algorithms break many classical crypto-systems. In response, people proposed post-quantum cryptography based on computational problems that are believed hard even for quantum computers. However, security of these schemes against \emph{quantum} attacks is elusive. This is because existing security analysis (almost) only deals with classical attackers and arguing security in the presence of quantum adversaries is challenging due to unique quantum features such as no-cloning. 

This work proposes a general framework to study which classical security proofs can be restored in the quantum setting. Basically, we split a security proof into (a sequence of) classical security reductions, and investigate what security reductions are ``quantum-friendly''.  We characterize sufficient conditions such that a classical reduction can be ``lifted'' to the quantum setting. 

We then apply our lifting theorems to post-quantum signature schemes. We are able to show that the classical generic construction of hash-tree based signatures from one-way functions and and a more efficient variant proposed in~\cite{BDH11} carry over to the quantum setting. Namely, assuming existence of (classical) one-way functions that are resistant to efficient quantum inversion algorithms, there exists a quantum-secure signature scheme. We note that the scheme in~\cite{BDH11} is a promising (post-quantum) candidate to be implemented in practice and our result further justifies it. Actually, to obtain these results, we formalize a simple criteria, which is motivated by many classical proofs in the literature and is straightforward to check. This makes our lifting theorem easier to apply, and it should be useful elsewhere to prove quantum security of proposed post-quantum cryptographic schemes. Finally we demonstrate the generality of our framework by showing that several existing works (Full-Domain hash in the quantum random-oracle model~\cite{Zha12ibe} and the simple hybrid arguments framework in~\cite{HSS11}) can be reformulated under our unified framework.  

\ftodo{Final sentence not satisfying}

\end{abstract}

\section{Introduction}
\label{sec:intro}

Advances in quantum information processing and quantum computing have brought about fundamental challenges to cryptography. Many classical cryptographic constructions are based on computational problems that are assumed hard for efficient classical algorithms. However, some of these problems, such as factoring, discrete-logarithm and Pell's equation, can be solved efficiently on a quantum computer~\cite{Shor97,Hallgren07}. As a result, a host of crypto-systems, e.g, the RSA encryption scheme that is deployed widely over the Internet, are broken by a quantum attacker. 

A natural countermeasure is to use \emph{quantum-resistant} assumptions instead. Namely, one can switch to other computational problems which appear hard to solve even on quantum computers, and construct cryptographic schemes based on them. Examples include problems in discrete lattices~\cite{MR09,Pei09} and hard coding problems~\cite{Sen11}. We can also make generic assumptions such as the existence of one-way functions that no efficient quantum algorithms can invert. This leads to the active research area termed \emph{post-quantum} cryptography~\cite{BBD09}. 
Nonetheless, quantum-resistant assumptions alone do not immediately imply quantum security of a scheme, due to other fundamental issues that could be subtle and easily overlooked.


First of all, we sometimes fail to take into account possible attacks unique to a quantum adversary 
. In other words, classical definition of security may not capture the right notion of security in the presence of quantum attackers\footnote{Although our focus is security against computationally bounded attackers, this issue is also relevant in the \emph{statistical} setting. There are classical schemes, which are proven secure against unbounded classical attackers, broken by attackers using quantum entanglement~\cite{CSST11}.}. For example, many signature schemes are designed in the random-oracle (RO) model, where all users, including the attacker, can query a truly random function. This is meant to capture an idealized version of a hash function, but in practice everyone instantiate it by him/herself with a concrete hash function. As a result, when we consider quantum attacks on these schemes, there seems no reason not to allow a quantum adversary to query the random-oracle in quantum superposition. This leads to the so called quantum random-oracle model~\cite{BDF+11}, in which we need to reconsider security definitions (as well as the analysis consequently)~\cite{Zha12ibe,Zha12prf,BZ13}. 

A more subtle issue concerns security proofs, which may completely fall through in the presence of quantum attacks. Roughly speaking, one needs to construct a \emph{reduction} showing that if an efficient attacker can successfully violate the security requirements of a scheme then there exists an efficient algorithm that breaks some computational assumption. However, a classical reduction may no longer work (or make sense at all) against quantum adversaries. A key classical technique, which encounters fundamental difficulty in the presence of quantum attackers, is called \emph{rewinding}. Loosely speaking, rewinding arguments consist of a mental experiment in which an adversary for a scheme is executed multiple times using careful variations on its input. This usually allows us to gain useful information in order to break the computational assumption. As first observed by van de Graaf~\cite{Graaf97}, rewinding seems impossible with a quantum adversary since running it multiple times might modify the entanglement between its internal state and an outside reference system, thus changing the system's overall behavior. This issue is most evident in cryptographic protocols for zero-knowledge proofs and general secure computation. There has been progress in recent years that develops quantum rewinding techniques in some special cases~\cite{Wat09,Unr12pok}, and a few classical protocols are proven quantum-secure~\cite{DL09,LN11,HSS11}. Hallgren et al.~\cite{HSS11} also formalized a family of classical security proofs against efficient adversaries that can be made go through against efficient quantum adversaries under reasonable computational assumptions. Despite these efforts, however, still not much is known in general about how to make classical security proofs go through against quantum attackers.  

This note revisits these issues for post-quantum cryptography based on computational assumptions, focusing on simple \emph{primitives} such as signatures, encryptions and identifications, where constructions and analysis are usually not too complicated (compared to secure computation protocols for example). In this setting, the issues we have discussed seem less devastating. For instance, rewinding arguments appear only occasionally, for example in some lattice-based identification schemes~\cite{Lyu08,Lyu09}. Usually rewinding is not needed for the security proof. Nonetheless, it is still crucial to pinning down proper security definitions against quantum attacks, as illustrated in the quantum random-oracle example above. In addition, just because there are no rewinding arguments, does not mean that we can take for granted that the security reduction automatically holds against quantum attackers. Very often in the literature of post-quantum cryptography, a construction based on some quantum-resistant assumption is given together with a security proof for \emph{classical} attackers only. The construction is then claimed to be quantum-secure without any further justification. In our opinion, this is not satisfying and quantum security of these schemes deserves a more careful treatment. 


\mypar{Contributions} The main contribution of this note is a general framework to study which classical security proofs can be restored in the quantum setting. A security proof can be split into (a sequence of) classical security reductions, and we investigate what reductions are ``quantum-friendly''. Recall that informally a reduction transforms an adversary of one kind to another. We distinguish two cases, \emph{game-preserving} and \emph{game-updating} reductions. 

A game-preserving reduction is one such that the transformation still makes sense, i.e., syntacticly  well-defined, for quantum adversaries. In this case we propose the notion of \emph{class-respectful} reductions which ensures in addition that the adversary obtained from the transformation indeed works (e.g., it is an efficient quantum algorithm and successfully solves some problem). Motivated by the structure of security reductions that occur in many post-quantum cryptographic schemes, we further characterize a simple criteria, which is straightforward to check. This makes the lifting theorem easier to apply, and should be useful to prove quantum security for many other schemes not restricted to the applications we show later in this note.  

On the other hand, a {game-updating} reduction captures the case that the classical reduction no longer makes sense, as illustrated by the quantum random-oracle model. This is usually more difficult to analyze. We propose \emph{translatable} reductions, which essentially reduces the problem to the game-preserving case. The basic idea is to introduce an ``interpreter'', so that the classical reduction becomes applicable to a quantum adversary with the translation by the interpreter. In both cases, we show in our lifting theorems that a reduction can be constructed in the quantum setting if there is a classical reduction that is {respectful} or {translatable} respectively. 

We apply our framework to prove quantum security of some hash-based signature schemes. Specifically, we show that the classical generic construction of hash-tree based signature schemes from one-way functions carries over to the quantum setting, assuming the underlying one-way function is quantum-resistant. This is also true for a more efficient variant proposed in~\cite{BDH11} assuming quantum-resistant pesudorandom functions, which in turn can be based on quantum-resistant one-way functions from known results. This scheme is a promising (post-quantum) candidate to be implemented in practice and our result further justifies it. Moreover, we give an alternative proof for the security of a general construction of signatures based on trapdoor permutations called Full-Domain hash in the quantum random-oracle model. We also show that an existing framework in the context of cryptographic protocols that characterizes a class of ``quantum-friendly'' classical security proofs (simple hybrid augments~\cite{HSS11}) fits our framework. These demonstrate the generality of our framework. 

\mypar{Remarks} Our framework (e.g., definitions of games and reductions) should look natural to people familiar with the provable-security paradigm. It should also be straightforward (or even obvious for experts) to verify the characterizations of ``quantum-friendly'' reductions in our lifting theorems. The purpose of this note, however, is to at least make people become more serious and cautious and to encourage further research, in addition to suggesting one possible formal framework to reason about the security of post-quantum cryptography against quantum attacks. Likewise, it may be just a tedious exercise to work though the classical proof for hash-based signatures and convince oneself it is indeed quantum-secure. Nonetheless, this can be done in a more abstract and rigorous way using our framework. We hope that our framework can be applied elsewhere to analyze quantum security of other post-quantum cryptographic constructions. Ideally, in some easy cases, it would serve as a tool to automate the routine parts, so that whoever designs a new scheme should be able to make some simple checks and conclude its quantum security.

\mypar{Other Related Works} There are a few works that study systematically what classical proofs or statements can be ``lifted'' to the quantum setting in the context of multi-party secure computation. Unruh in~\cite{Unruh10} showed that any classical protocol that is secure in the statistical setting, i.e., against computationally \emph{unbounded} adversaries, under a strong \emph{universal-composable} notion is also statistically secure in an analogous quantum universal-composable model. Fehr et al.~\cite{FKSZZ13} considered \emph{reducibility} between two-party cryptographic tasks in the quantum setting. For example, one can ask if there is a secure protocol for oblivious transfer assuming two parties can perform bit commitments securely. They showed that in most cases, the reducibility landscape remains unchanged in the quantum setting under the very same classical protocols. However, there are cases that classical reducibility no longer holds quantumly, and sometimes new relations can be established using quantum protocols. 

The formalization of games, reductions and other terms in this note is influenced by a lot of classical literatures on game-playing proofs~\cite{GM84,Yao82,KR01,BR06,Shoup05,Halevi05}. Recent developments, especially the framework of code-based game-playing proofs~\cite{Halevi05,BR06} have motivated automated tools for proving security~\cite{Blanchet08,BGZ09,Stump09,BGHZ11}. Our treatment of computational assumptions is also inspired by the line of works classifying complexity-theoretic intractability assumptions~\cite{Naor03,HH09,Pass11,GW11}.

\section{Preliminary}
\label{sec:prelim}

\mypar{Basic notations} For $m \in  \natural$, $[m]$ denotes the set $\{1, \ldots, m\}$. We
use $n \in \mathbb{N}$ to denote a {\em security parameter}. The
security parameter, represented in unary, is an implicit input to
all cryptographic algorithms; we omit it when it is clear from the
context. Quantities derived from protocols or algorithms
(probabilities, running times, etc)  should be thought of as
functions of $n$, unless otherwise specified. A function $f(n)$ is
said to be negligible if $f = o(n^{-c})$ for any constant $c$, and
$\negl(n)$ is used to denote an unspecified function that is
negligible in $n$. We also use $poly(n)$ to denote an unspecified
function $f(n) = O(n^c)$ for some constant $c$. When $D$ is a
probability distribution, the notation $x \gets D$ indicates that
$x$ is a sample drawn according to $D$. When $D$ is a finite set, we
implicitly associate with it the uniform distribution over the set.
If $D(\cdot)$ is a probabilistic algorithm, $D(y)$ denotes the
distribution over the output of $D$ corresponding to input $y$. We
will sometimes use the same symbol for a random variable and for its
probability distribution when the meaning is clear from the context.
Let $\mathbf{X} = \{X_n\}_{n\in \natural}$ and $\mathbf{Y} =
\{Y_n\}_{n \in \natural}$ be two ensembles of binary random
variables. We call $\mathbf{X,Y}$ {\em indistinguishable}, denoted
$\mathbf{X} \approx \mathbf{Y}$, if $ \left|\Pr(X_n = 1) - \Pr(Y_n =
1)\right| \leq \negl(n)$.

\mypar{Machine Models} We model classical parties as interactive Turing machines, which are probabilistic polynomial-time ({\sf PPT}) by default. 
Quantum machines are modelled following that of~\cite{HSS11}.  A
{\em quantum interactive machine} (QIM) $M$ is an ensemble of
interactive circuits $\{M_n\}_{n\in \natural}$. For each value $n$
of the security parameter, $M_n$ consists of a sequence of circuits
$\{M_n^{(i)}\}_{i=1,...,\ell(n)}$, where $M_n^{(i)}$ defines the
operation of $M$ in one round $i$ and $\ell(n)$ is the number of
rounds for which $M_n$ operates (we assume for simplicity that
$\ell(n)$ depends only on $n$). We omit the scripts when they are
clear from the context or are not essential for the discussion. $M$
(or rather each of the circuits that it comprises) operates on three
registers: a state register {\sf S} used for input and workspace; an
output register {\sf O}; and a network register {\sf N} for communicating
with other machines. 
The size (or running time) $t(n)$ of $M_n$ is the sum of the sizes
of the circuits $M_n^{(i)}$. We say a machine is polynomial time if
$t(n)=poly(n)$ and
there is a deterministic classical Turing machine that computes the
description of $M_n^{(i)}$ in polynomial time on input $(1^n,1^i)$.
When two QIMs $M$ and $M'$ interact, their network register {\sf N} is
shared. The circuits $M_n^{(i)}$ and ${M'}_n^{(i)}$ are executed
alternately for $i=1,2,...,\ell(n)$. When three or more machines
interact, the machines may share different parts of their network
registers (for example, a private channel consists of a register
shared between only two machines; a broadcast channel is a register
shared by all machines). The order in which machines are activated
may be either specified in advance (as in a synchronous network) or
adversarially controlled.


\section{Defining Games and Reductions}
\label{sec:defs}
This section introduces a formal definition of reductions, which captures the type of security reductions that we care mostly about. It builds upon a basic notion of 
games. 

We use \emph{game} $G$ to denote a general probabilistic process between two players: the challenger $\chr$ initiates the interaction with the other player, call it an adversary $\adv$. After several rounds of communication, $\chr$ outputs one bit $\succ$ or $\fail$ indicting success or failure of the game. We define the game value of $G$ with an adversary $\adv$ to be the probability that $\chr$ outputs $\succ$, and denote it $\gval{G}(\adv)$. Typically in a game $G$, $\chr$ is efficient, i.e., a poly-time classical or quantum machine. Very often we want to analyze the game when the adversary is restricted to a class of machines $\gcls$ (e.g., poly-time classical machines). We write $G(\gcls)$ to indicate this case, and define $\gval{\game}(\gcls): = \text{max} \{ \gval{\game}(\adv): {\adv\in\gcls}\}$.  Sometimes we denote $\qgame$ to stress a game defined for quantum machines. 
We describe below as an example the standard forgery game of existentially unforgeable signatures under (adaptive) chosen-message-attacks (\eucma)~\cite{GMR88,KL07}. 

\begin{ncprot}{Existential-Forgery Game}{$\gfor$}
{\sf Signature scheme}: $\Pi = (\kgen,\sign,\vrfy)$.
\begin{itemize}
	\item $\chr$ generates $(pk,sk) \gets \kgen(1^n)$. Send $pk$ to adversary $\adv$. 
	\item $\adv$ can query signatures on messages $\{m_i\}$. $\chr$ returns $\sigma_i : = \sign(sk,m_i)$. These messages can be chosen adaptively by $\adv$. 
	\item $\adv$ outputs $(m^*,\sigma^*)$. If $\vrfy(pk,(\sigma^*,m^*)) = 1$ and $m^*\notin\{m_i\}$, $\chr$ outputs $\succ$. Otherwise output $\fail$.  
\end{itemize}
\end{ncprot}

There are many variants of this game which will be used later in this note. For example, we denote the game in which $\adv$ is allowed to query at most one signature $\gforot$. $\gforro$ denotes the game where a random-oracle is available to both parties, and if the random-oracle can be accessed in quantum superposition we denote the game $\gforqro$. 


We define a reduction $\red$ as a 3-tuple $(\extg, \trr, \intg)$. There are an external (explicit) game $\extg$ and an internal (implicit) game $\intg$, and an additional party $\trr$ called the \emph{transformer}. 
Loosely speaking, $\trr$ transforms an adversary $\adv$ in $\intg$ to an adversary in $\extg$.  
Specifically, $\trr$ takes an adversary's machine $\adv$ as input and outputs the description of an adversary in $\extg$. We distinguish \emph{black-box} and \emph{non-black-box} reductions, with a focus on black-box reductions. In a black-box reduction, $\adv$ is provided as a black-box, which means that the transformation does not look into the codes and inner workings of the adversary. Whereas in a non-black-box reduction, $\red$ has the explicit description of $\adv$. We denote $\trr(\adv)$ as the resulting adversary in $\extg$ that is ``transformed'' from $\adv$ by $\trr$. In the black-box setting, the output of $\trr$ will always be of the form $\btrr^{\adv}$, i.e., an oracle machine with access to $\adv$. Note that $\btrr$ is the same for all $\adv$, and it emulates an execution of $\intg$ with $\adv$. However, in general $\btrr$ needs not to run the game as in a real interaction. For instance, it can stop in the middle of the game and start over (i.e., rewind).   

\fnote{insert a picture}

\mypar{Properties of a reduction} 
To make a reduction meaningful, we describe below a few properties that we may want a reduction to hold. Let $\clsa$ and $\clsb$ be two classes of machines. 

\begin{itemize}
	\item \textbf{$\clsa$-compatible} reductions. We say $\red$ is $\clsa$-compatible, if $\forall \adv\in\clsa$, $\intg(\adv)$ and $\extg(\tadv{\trr}{\adv})$ are well defined. Namely $\clsa$ and $\trr(\adv)$ respect the specifications of the games.  
	\item \textbf{$(\clsa,\clsb)$-consistent} reductions. 
We say $\red$ is $(\clsa,\clsb)$-consistent, if $\red$ is $\clsa$-compatible and $\forall \adv \in \clsa$, $\trr(\adv) \in \clsb$. 
 When we write a reduction as $(\extg(\clsb), \trr, \intg(\clsa))$ or $\red(\clsa,\clsb)$ in short, the reduction is assumed to be $(\clsa,\clsb)$-consistent. Note that if $\red$ is black-box, it must hold that $\btrr^\clsa\subseteq\clsb$.
	\item \textbf{Value-dominating}. We say $\red$ is {value-dominating} if $\gval{\extg}(\tadv{\trr}{\adv}) = \gval{\extg}(\tadv{\trr}{\badv})$ whenever $\gval{\intg}(\adv) = \gval{\intg}(\badv)$. 
	\item \textbf{$(\asucc,\clsa)$-effective} reductions. Let $\asucc: \preal \to \preal$ be some function. We say $\red$ is $\asucc$-effective on $\adv$ if $\gval{\extg}(\trr(\adv)) \geq \asucc(\gval{\intg}(\adv))$. If this holds for any $\adv\in\clsa$, we call $\red$ $(\asucc,\clsa)$-effective
	\item \textbf{($\atime,\clsa$)-efficient} reductions.  Let $\atime: \preal \to \preal$ be some function. We say $\red$ is $\atime$-efficient if $\rtime{\trr(\adv)} \leq \atime(\rtime{\adv})$ for any $\adv\in \clsa$.  
	
\end{itemize}
Effective and efficient reductions are often used in combination, especially when we are concerned with tightness of a reduction. In that case, $\asucc$ and $\atime$ may depend on both $\rtime{\adv}$ and $\gval{\intg}(\adv)$.  This paper will focus on effectiveness only. We often abuse notation and use $\asucc$ as a scalar if this causes no confusion. We stress that these properties talk about the output machine of $\trr$ on $\adv$ (e.g., $\trr(\adv)$ lies in a specific class, or it runs in time comparable to that of $\adv$),  however we do not restrict the computational power of $\trr$, though it is typically efficient. The reason is that for our purpose, we only need to show existence of an adversary for $\extg$ with nice properties.

\section{Quantum-Friendly Security Reductions: A General Framework}
\label{sec:qfred}

In this section, we attempt to propose a general framework to study which classical proofs still hold when the adversaries become quantum. Consider a classical cryptographic scheme $\Pi$. To analyze its security against efficient classical attacks (in the provable-security paradigm), one typically proceeds as follows: 

\begin{enumerate}
	\item Formalizing some security requirement by a game $\intg$. Typically we are concerned about security against a particular class of attackers (e.g., PPT machines), so we restrict the game $\intg$ to a class $\clsa$. We also associate a value $\veps_\clsa \in (0,1]$ with the game, which upper bounds the success probability that any adversary in $\clsa$ wins the game. Namely we require that $\gval{\intg}(\clsa) \leq \veps_\clsa$. We denote this security requirement as $(\intg(\clsa),\veps_\clsa)$\footnote{Sometime we write $(\intg(\clsa),\veps_\clsa)_\Pi$ to emphasize the specific scheme we are dealing with, though it is usually clear from the context.}. 
	\item Formalizing some computational assumption 
 by another game $\extg$. Similarly the assumption is assumed to hold against a specific class of machines, so we restrict the game to a class $\clsb$, and require that $\gval{\extg}(\clsb) \leq \secpar{\clsb} \in (0,1]$. Denote the computational assumption as $(\extg(\clsb),\secpar{\clsb})$. 
	\item Constructing an $(\clsa,\clsb)$-consistent reduction $\red=(\extg(\clsb), \trr, \intg(\clsa))$. Security follows if the reduction is in addition $\asucc$-effective with $\asucc \geq \secpar{\clsb} /{\secpar{\clsa}}$. This implies if there exists an $\adv\in\clsa$ with $\gval{\intg}(\adv) > \secpar{\clsa}$ (i.e.. $\adv$ breaks the security requirement), there is an adversary $\trr(\adv) \in \clsb$ such that $\gval{\extg}(\tadv{\trr}{\adv}) \geq \asucc\cdot\gval{\intg}(\adv)  > \secpar{\clsb}$ (i.e., it breaks the computational assumption). 
\end{enumerate}

Now we want to know if the classical security reductions are ``quantum-friendly'' so that we can claim that the scheme is secure against quantum attacks. 
We need to reconsider each step of the classical analysis in the quantum setting (See Table~\ref{tab:compqps} for a comparison between classical provable-security and quantum provable-security for a scheme.). Let $(\qclsa,\qclsb)$ be two classes of quantum machines.  We adapt $\intg$ and define $(\intgq(\qclsa),\secpar{\qclsa})$. It is supposed to 
capture some security requirement against quantum attackers in $\qclsa$, and we require that $\gval{\intgq}(\qclsa) \leq \secpar{\qclsa}$. Likewise, we adapt $\extg$ to a game $\extgq$, which should formalize a reasonable computational assumption $(\extgq(\qclsb), \secpar{\qclsb})$ against quantum adversaries. Then we can ask the fundamental question (still informal): 

\begin{center}
\emph{Can we ``lift'' $\red$ to the quantum setting?\\
 Namely, is there a reduction $\qred(\qclsa,\qclsb)$ that preserves similar properties as $\red(\clsa,\clsb)$}? 
\end{center}

To answer this question, we distinguish two cases. In the simpler case, $\qgame$ are syntactically identical to $\game$. Namely, $\extgq(\qclsb)$ (resp. $\intgq(\qclsa)$) is just $\extg$ (resp. $\intg$) restricted to the quantum class $\qclsb$ (resp. $\qclsa$). In particular, this means that $\extg$ and $\intg$ are still the right games that capture a computational assumption and some security requirement. We call this case \emph{game-preserving}. In contrast, as illustrated by the quantum random-oracle example, $\qgame$ may change and this leads to a more complicated case to analyze. We call it \emph{game-updating}. In the following subsections, we investigate in each case what reductions can be lifted to the quantum setting, and hence are quantum-friendly. 

\begin{table}[htdp]
\caption{Components of classical and quantum provable-security for a classical construction.}
\begin{center}
\begin{tabular}{|c|c|c|}
\hline
& Classical Provable-Security & Quantum Provable-Security \\
\hline
Security Requirement & $(\intg(\clsa),\secpar{\clsa})$ &  $(\intgq(\qclsa), \secpar{\qclsa})$ \\
\hline
Computational Assumption & $(\extg(\clsb), \secpar{\clsb})$ &  $(\extgq(\qclsb), \secpar{\qclsb})$ \\
\hline
Reduction & $\red(\clsa,\clsb)$  & $\stackrel{?}{\longrightarrow} \qred(\qclsa,\qclsb)$ \\
\hline
\end{tabular}
\end{center}
\label{tab:compqps}
\end{table}%

\subsection{Lifting Game-Preserving Reductions}
\label{ssec:gpr}

Let $\red(\clsa,\clsb)=\sredcls$ be a classical reduction. Let $\extgq(\qclsb)$ and $\intgq(\qclsa)$ be extended games in the quantum setting that are restricted to classes of quantum machines $\qclsb$ and $\qclsa$. We consider the case that $\qgame$ and $\game$ are the same in this section. We want to know if there is a reduction $\qred(\qclsa,\qclsb)$ that preserves nice properties of $\red$. Since we are dealing with the same games applied to different classes of machines, one may expect that simple tweaks on $\red$ should work. This intuition is indeed true to some extend, which we formalize next. 

\begin{definition}[$G$-equivalent machines] 
Two machines $M$ and $N$ are called $G$-equivalent if $\gval{G}(M) \equiv \gval{G}(N)$.  
\label{def:respred}
 \end{definition} 
 
\begin{definition}[$\geqm{G}{\gcls}$-realizable classical machines] 
A classical machine $M$ is called $\geqm{G}{\gcls}$-\realize, if there is a machine $N \in \gcls$ s.t. $\gval{G}(M) = \gval{G}(N)$. We denote $\eqmcls{G}{\gcls}$ as the collection of classical machines that are $\geqm{G}{\gcls}$-\realize. 
\label{def:respred}
 \end{definition} 

We put forward \emph{class-respectful} reductions as a template for quantum-friendly reductions in the game-reserving case. 

\begin{definition}[$\beta$-$(\qclsa,\qclsb)$-respectful reductions]
Let $\red$ be a classical reduction \\
$\sredcls$. We say $\red$ is $\beta$-$(\qclsa,\qclsb)$-respectful for some $\beta\in\preal$ if the following hold:

\begin{enumerate}
	\item \textbf{$(\beta,\qclsa)$-extendable}: $\red$ is $\eqmcls{\intg}{\qclsa}$-compatible and $(\beta,\eqmcls{\intg}{\qclsa})$-effective. That is, 
$\forall \adv \in \eqmcls{\intg}{\qclsa}$, $\extg(\tadv{\trr}{\adv})$ and $\intg(\adv)$ are well-defined\footnote{Most classical games we deal with are actually well-defined for all machines. But we explicitly state this requirement in case of some artificial examples.}, and $\gval{\extg}(\tadv{\trr}{\adv}) \geq \beta (\gval{\intg}(\adv))$. 
	\item \textbf{$(\qclsa,\qclsb)$-closed}: $\red$ is $(\eqmcls{\intg}{\qclsa},\eqmcls{\extg}{\qclsb})$-consistent. Namely, $\forall \adv \in E_{\intg}(\qclsa)$, $\tadv{\trr}{\adv} \in 
\eqmcls{\extg}{\qclsb}$. 
\end{enumerate}
\label{def:respred}
 \end{definition} 
 
 The theorem below follows (almost) immediately from this definition. 
 
 \begin{theorem}[Quantum lifting for game-preserving reductions]
If $\red(\clsa,\clsb)$ is $\beta$-$(\qclsa,\qclsb)$-respectful, then there exists an $(\qclsa,\qclsb)$-consistent reduction ${\qred}(\qclsa,\qclsb) :=$ \\
$(\extg(\qclsb),\qtrr,\intg(\qclsa))$ that is 
$(\beta,\qclsa)$-effective.
 \label{thm:ql-gpr}
 \end{theorem}
 
 \begin{proof}
Consider any $\qadv\in\qclsa$. Let $\adv$ be a classical machine such that $\adv$ is $\extg$-equivalent to $\qadv$. Since $\red$ is $(\qclsa,\qclsb)$-closed, we know that $\tadv{\trr}{\adv} \in \eqmcls{\extg}{\qclsb}$ and hence there is a machine $N_\qadv \in\qclsb$ s.t. $\gval{\intg}(N_\qadv) = \gval{\intg}(\tadv{\trr}{\adv})$. Define $\qver{\trr}$ to be a quantum machine such that, given $\qadv \in\qclsa$, outputs $N_\qadv$. Namely $\tadv{\qtrr}{\qadv} := N_\qadv$. Let $\qred: = (\extg(\qclsb), {\qver{\trr}}, \intg(\qclsa))$. Clearly $\qred$ is $(\qclsa,\qclsb)$-consistent due to the way we defined $\qtrr$.  It is also $(\beta, \qclsa)$-effective because $\gval{\extg}(\tadv{\qtrr}{\qadv})  = \gval{\extg}(\tadv{\trr}{\adv})\geq \beta(\gval{\intg}(\adv)) = \beta(\gval{\intg}(\qadv))$. 
 \end{proof}
 

To apply the theorem, we need to check the two conditions of respectful reductions. The ``extendability'' condition is usually easy to verify. However, the  ``closure'' property can be challenging and subtle, depending on the classes of players we care about. We will be mostly interested in poly-time machines. Namely let $\clsa = \clsb$ be {poly-time classical machines} and $\qclsa = \qclsb$ be the collection of {poly-time quantum machines}, denote it by $\qpoly$. In this case, we propose a simple criteria that is easy to check in existing classical security reductions. When combined with a few other easily verifiable conditions, we can show class-respectful reductions. This in a way justifies a common belief that most post-quantum schemes are indeed quantum-secure, due to some simple form in their classical security reductions which seem ``quantum-friendly''. 

Let $\red=(\extg, \trr,\intg)$ be a classical black-box reduction. We say that $\red$ is \emph{straight-line} if the output machine of $\trr$ on $\adv$, which as before is denoted $T^\adv$, runs $\adv$ in straight-line till completion. Namely, other than the flexibility of choosing $\adv$'s random tape, $T$ behaves exactly like a honest challenger in $\intg$ when it invokes $\adv$. This type of reduction, due to its simple structure, is amenable to getting lifted. 

\begin{theorem}[Straight-line reduction: a useful condition for class-closure]
Let $\red = (\extg(\clsb), \trr, \intg(\clsa))$ be a classical reduction with $\clsa$ and $\clsb$ both being classical poly-time machines. Let $\qclsa = \qclsb$ be quantum poly-time machines. If $\red$ is black-box straight-line, $\qclsa$-compatible and value-dominating, then $\red$ is $(\qclsa,\qclsb)$-closed. 
\label{thm:slr}
\end{theorem}

\begin{proof}
For any $\adv\in\eqmcls{\intg}{\qclsa}$, let $\qadv\in\qclsa$ be such that $\adv$ and $\qadv$ are $\intg$-equivalent.  We argue that $T^\adv$  and $T^\qadv$ are $\extg$-equivalent and hence $T^\adv\in\eqmcls{\extg}{\qclsb}$.  Since $\adv$ and $\qadv$ are $\intg$-equivalent and $\red$ is value-dominating, $\gval{\extg}(T^\qadv) = \gval{\extg}(T^\adv)$. $T^\qadv \in \qclsb$, i.e., it is quantum poly-time, since $T$ is classical poly-time, and runs any oracle in straight-line. Finally, note that we need the compatibility condition so that all objects above are well-defined. 
\end{proof}

Combine the extendibility condition, we get the corollary below from Theorem~\ref{thm:ql-gpr}. 

\begin{corollary}
Let $\red$ be a classical black-box reduction for classical poly-time players. 
Let $\qclsa=\qclsb$ be quantum poly-time machines. If $\red$ is  $(\beta,\qclsa)$-extendible, straight-line, and value-dominating, then $\red$ is $\respab$-respectful. As a consequence, there is a reduction $\qred(\qclsa,\qclsb)$ that is $(\beta,\qclsa)$-effective.
\label{cor:sltoresp}
\end{corollary}

Note that in this scenario, $\qred$ is also straight-line and $\qtrr(\qadv) = T^\qadv$. Loosely speaking, the very same reduction carries over to the quantum setting.


\subsection{Lifting Game-Updating Reductions}
\label{ssec:gur}

Sometimes we need to update $\extgq$ or $\intgq$ or both, in order to capture the right computational assumption and the security property against quantum players. In this case, the classical transformation procedure may become totally inappropriate and give little clue about how to restore a quantum reduction (if there exists one). 

 \def\qtnod{{\qtrr_0}}

We view this issue as a matter of ``language-barrier''. One way to establish a reduction $\qred(\qclsa,\qclsb)$ is to introduce an \emph{interpreter} $\qintp$ that translates the ``languages'' between the players in the original (classical) and updated (quantum) games. Namely,  $\qintp$ translates an adversary $\qadv$ in $\intgq$ to an adversary $\qadv'$ in the classical game $\intg$. Then we can reduce the issue to the game-preserving case and consider a class of quantum adversaries $\qclsa':=\qintp(\qclsa)$. Suppose we can lift the classical reduction to work with adversaries in $\qclsa'$, then we end up with a quantum adversary in game $\extg$. 
Next, by the same token, $\qintp$ translates the adversary into a quantum one compatible in $\extgq$. This procedure gives a quantum transformer $\qtrr: = \qintp \circ \qtnod \circ \qintp$ that operates as follows 
$$ \qadv \in \qclsa \stackrel{\qintp}{\longrightarrow} \qadv' \stackrel{\qtnod}{\longrightarrow} \tadv{\qtnod}{\qadv'} \stackrel{\qintp}{\longrightarrow} \bqadv \in \qclsb \, .$$

\ftodo{insert picture with interpreter}

We formalize this idea, and propose \emph{class-translatable} reductions as a template for quantum-friendly reductions in the game-updating case. For simplicity, we assume only $\intg$ is updated to $\intgq$ and $\extg$ stays the same
. We want to investigate if a reduction of the form $(\extg(\qclsb),\qtrr,\intgq(\qclsa))$ can be derived. It is straightforward to adapt the treatment to the scenario where $\extg$ (or both) gets updated. 

\begin{definition}[$(\beta,\beta')$-$(\qclsa,\qclsb)$-translatable reductions]
Let $\red$ be a classical reduction $\sredcls$ and $\beta, \beta'$ be two functions. Let $\intgq$ be a quantum game, and $(\qclsa,\qclsb)$ be classes of quantum machines. We say $\red$ is \transred-translatable, if there exists a machine (i.e. Interpreter) $\qintp$, such that the following hold: 
\begin{itemize}
	\item $\red$ is $\beta$-$(\qclsb,\qclsa')$-respectful, where $\qclsa': = \qintp(\qclsa)$. 
	\item $(\intg, \qintp, \intgq)$ is a $(\beta',\qclsa)$-effective reduction. Namely $\forall \qadv\in \qclsa$, $\gval{\intg}(\qintp(\qadv)) \geq \beta'(\gval{\intgq}(\qadv))$. 
\end{itemize}

\end{definition}

 \begin{theorem}[Quantum lifting for game-updating reductions]
If $\red(\clsa,\clsb)$ is \transred-translatable, then there exists an $(\qclsa,\qclsb)$-consistent reduction ${\qred}(\qclsa,\qclsb) := (\extg(\qclsb),\qtrr,\intgq(\qclsa))$ that is 
$(\beta\cdot\beta',\qclsa)$-effective.
 \label{thm:ql-gur}
 \end{theorem}
 
 \begin{proof}
By the hypothesis, we know there is an interpreter $\qintp$. Since $\red$ is \respred{\qclsb}{\qclsa'}-respectful, by Theorem~\ref{thm:ql-gpr}, there is a $\qtrr_0$ s.t. $(\extg(\qclsb), \qtnod,\intg(\qclsa'))$ is $(\beta,\qclsa')$-effective. Define $\qtrr : = \qtnod \circ \qintp$ and $\qred: = (\extg,\qtrr,\intgq)$. Clearly, $\qred$ is $(\qclsa,\qclsb)$-consistent because for any $\qadv\in \qclsa$, $\qtrr(\qadv) = \qtnod(\qintp(\qadv)) \in \qclsb$. On the other hand, for any $\qadv\in \qclsa$ it holds that  $\gval{\extg}(\qtrr(\qadv)) \geq \beta \cdot \gval{\intg}(\qintp(\qadv)) \geq \beta \beta' \cdot \gval{\intgq}(\qadv)$.   
 \end{proof}
 
In contrast to the game-preserving setting, applying lifting theorem for game-updating reductions typically needs non-trivial extra work. The main difficulty comes from showing existence of an interpreter $\qintp$ with the desired properties. In Sect.~\ref{ssec:appgu-fdh}, we give an example that demonstrates potential applications of Theorem~\ref{thm:ql-gur}.  
 

\section{Applications}
\label{sec:app}

We give a few examples to demonstrate our framework for ``quantum-friendly'' reductions. In the game-preserving setting (Section~\ref{ssec:appgp}), we show two versions of quantum-secure hash-based signatures schemes assuming quantum-resistant one-way functions. One follows the generic construction that builds upon Lamport's OTS and Merkle's original hash-tree idea.  The other is an efficient variant proposed in~\cite{BDH11} that uses a more compact one-time signature scheme and a more sophisticated tree structure. In the game-updating setting (Section~\ref{ssec:appgu-fdh}), we give an alternative proof  for Full-Domain Hash (FDH) in the Quantum RO model as shown in~\cite{Zha12ibe}. We stress that this proof is meant to illustrate how our lifting theorem can be potentially applied, as apposed to providing new technical insights. Unless otherwise specified, all players are either classical or quantum poly-time machines. 
  
\subsection{Quantum Security for Hash-based Signatures}
\label{ssec:appgp}

Classically, there are generic constructions (and efficient variants) for \eucma-secure signature schemes based on one-way functions. We show that security reductions there can be lifted easily, using our \emph{class-respecful} characterization. It follows that there are classical signature schemes that are secure against quantum attacks, merely assuming existence of quantum-resistant one-way functions. 

\subsubsection{Generic hash-tree signature schemes}
\label{sssec:merkle}

A generic approach for constructing \eucma-secure signature scheme from $\owfs$ goes as follows: 

\begin{itemize}
	\item A one-time signature ($\ots$) is constructed based on $\owfs$. There are various ways to achieve it. We consider Lamport's construction (\lots) here~\cite{Lamport79}. 
	\item A family of universal one-way hash functions ($\uowhf$) 
	is constructed based on $\owfs$. This was shown by Rompel~\cite{Romp90} and we denote the hash family \ruowhf.
	\item An $\ots$ scheme is converted to a full-fledged (stateful) signature scheme using $\uowhf$.  The conversion we consider here is essentially Merkle's original hash-tree construction~\cite{Mer90}. 
\end{itemize}
We show next that each step can be ``lifted'' to the quantum setting using our lifting theorem for game-preserving reductions (Theorem~\ref{thm:ql-gpr}) and the straight-line characterization (Theorem~\ref{thm:slr}). Note that we do not intend to optimize the construction here. For instance, one can use a pseudorandom function to make the signature scheme stateless. Verifying whether these still hold in the quantum setting is left as future work, though we believe it is the case, following the framework and tools we have developed. 

\mypar{Lamport's $\ots$} Consider the (classical) reduction $\red:=(\ginv,\trr,\gforot)$, where $\ginv$ is the inversion game and $\gforot$ is the one-time forgery game. It is straight-line and value-dominating.  Both games are compatible with $\qpoly$ and $\gval{\extg}(T^{\adv}) \geq \beta \cdot \gval{\intg}(\adv)$ for any $\adv$ with $\beta (x) = \frac{1}{2\ell(n)} x$ and $\ell(n)$ a polynomial representing the length of the messages. Hence $\red$ is $(\beta,\qpoly)$-effective as well. Thus we claim that: 
\begin{proposition}
$(\ginv, \secpar{\qpoly} = \negl(n))_{\owf}$ implies $(\gforot, \secpar{\qpoly} = \negl(n))_{\lots}$.
Namely, assuming quantum-resistant $\owfs$, there exists $\eucma$-secure $\ots$ against quantum attackers $\qpoly$.
\label{prop:lots}
\end{proposition}

\mypar{$\uowhf$ from $\owfs$} Rompel's construction is complicated and the proof is technical (or rather tedious). However, the key ingredients in which security reductions are crucial are actually not hard to check. Basically, there are four major components in the construction: 

\begin{enumerate}
	\item From a given $\owf$ $f^0$, construct another  $\owf$ $f$ with certain structure. Basically, $f$ is more ``balanced'' in the sense that sampling a random element in the range of $f$ and then sub-sampling its pre-images is not much different from sampling a random element in the domain directly. 
	\item From $f$, construct $\uhash$ such that for any $x$, it is hard to find a collision in the so called ``hard-sibling'' set. The hard-sibling set should comprise a noticeable fraction of all possible collisions.  
	\item Amplifying the hard-sibling set so that finding any collision of a pre-determined $x$ is hard.
	\item Final refinements such as making the hash functions compressing.  
\end{enumerate}

The second step is the crux of the entire construction. There are three reductions showing that finding a hard-sibling is as hard as inverting $f$ which we will discuss in a bit detail below, whereas showing that the hard-sibling set is noticeably large is done by a probabilistic analysis and holds information-theoretically. Other steps either do not involve a security reduction and relies purely on some probabilistic analysis, or the reductions are clearly liftable. 

The three reduction in step 2 involve four games: $\ginv$--the standard inversion game for $\owfs$; $\ginvv$--a variant of $\ginv$ in which $y$ is sampled according to another distribution, as opposed to sampling a domain element $x$ uniformly at random and setting $y:=f(x)$; $\gcolv$, a variant of the collision game for $\uowhf$, in which an adversary is supposed to find a collision $x'$ in a special set (we don't specify it here); and $\gcolvv$, which further modifies $\gcolv$ in the distribution that $s$ is sampled (instead of uniformly at random). Then $\red_1=(\ginv,\trr_1,\ginvv)$, $\red_2:=(\ginvv, \trr_2,\gcolvv)$ and $\red_3 = (\gcolvv, \trr_3,\gcolv)$ are constructed. $\red_1$ and $\red_3$ essentially follow from the ``balanced'' structure of $f$, and $\red_2$ comes from the construction of $\uhash$. All three reductions are black-box straight-line, value-dominating, and $(\beta_i,\qpoly)$-effective with $\beta_i \geq 1/{p_i(n)} $ for some polynomial $p_i, i\in\{1,2,3\}$. {For concreteness, we can set $p_1 = \ell'(n)$--the length of the input string of $f^0$, $p_2 = 3$ a constant, and $p_3 (n) = 5\ell'(n) + \log \ell'(n) +2$. Our exposition here and parameter choices are adapted from~\cite{KK05}.} 

\begin{proposition}
$(\ginv, \secpar{\qpoly} = \negl(n))_\owf$ implies $(\gcol, \secpar{\qpoly} = \negl(n))_\ruowhf$.
Namely, assuming quantum-resistant $\owfs$, there exist $\uowhf$ secure against quantum attackers $\qpoly$.
\label{prop:luowhf}
\end{proposition}

\mypar{Hash-tree: converting $\ots$ to full-fledged signatures} Once a family of $\uowhf$ and an $\ots$ are at hand, we can get a full-fledged signature scheme based on Merkle's hash-tree construction. Basically, one constructs a depth-$k$ binary tree and each leaf corresponds to a message. Each node in the tree is associated with a key-pair $(pk_w, sk_w)$ of the $\ots$ scheme. The signature of a message $m$ consists of $\sigma_m:= \sign(sk_m, m)$ and an authentication chain. For each node $w$ along the path from the root to the message, we apply $\uhash$ to the concatenation of its children's public keys and then sign the resulting string with its secret key $sk_w$. The authentication chain contains all these $(pk_{w0}, pk_{w1},\sigma_w: = \sign(sk_w, pk_{w0}\|pk_{w1}))$. Let $\mtree$ be the resulting tree-based scheme and $\gfor$ be the forgery game. The classical security analysis builds upon two reductions $(\gcol,\trr,\gfor)$ and $(\gforot,\trr',\gfor)$. It is easy to check that both satisfy the conditions in Corollary~\ref{cor:sltoresp}.

\begin{proposition}
$(\gcol, \secpar{\qpoly} = \negl(n))_\uowhf$ and $(\gforot, \secpar{\qpoly} = \negl(n))_\ots$ imply $(\gfor, \secpar{\qpoly} = \negl(n))_\mtree$.
Namely, assuming quantum-resistant $\uowhf$ and $\ots$, there exist an $\eucma$-secure signature scheme against quantum attackers $\qpoly$.
\label{prop:lmtree}
\end{proposition}

Combining Propositions~\ref{prop:lots},~\ref{prop:luowhf}, and~\ref{prop:lmtree}, we get 
\begin{theorem}
Assuming quantum-resistant $\owfs$, there exists $\eucma$-secure signature schemes against quantum poly-time attackers $\qpoly$.
\label{thm:owftosig}
\end{theorem}

\subsubsection{XMSS: an efficient variant}
\label{sssec:xmss}

The \xmss~scheme~\cite{BDH11} can be seen an efficient instantiation of the generic construction above. It uses a different one-time signature scheme called Winternitz-OTS ($\wots$ for short), which can be based on a family of pseudorandom functions, which in turn exists from the ``minimal'' assumption that $\owfs$ exist. The hash-tree (which is called $\xmss$-tree in~\cite{BDH11}) also differs slightly. We now show that both the security of $\wots$ and the conversion by $\xmss$-tree are still valid against quantum adversaries. 

\mypar{Quantum security of $\wots$} Classically, existence of $\owf$ imply the $\eucma$-security of $\wots$. This is established in three steps: 1) By standard constructions, a pseudorandom generator (\prg) can be constructed from $\owfs$~\cite{HILL99}, and then one can construct a pseudo-random function ($\prf$) from a $\prg$~\cite{GMM86}. 2) A $\prf$ is shown to be also key-one-way ($\kow$, defined later). 3) Show that $\kow$ implies $\eucma$-security of $\wots$  by a reduction. 

The first step is known to be true in the presence of quantum adversaries~\cite{Zha12prf}\footnote{It is easy to verify that the security reduction from $\prg$ to $\prf$ in GMM construction is quantum friendly. The security analysis in the HILL $\prf$ construction from $\owfs$ is much more complicated. To the best of our knowledge, no rigorous argument has appeared in the literature. It would be a nice exercise to apply our framework and give a formal proof.}. Informally the game for $\kow$ of a function family $F$ goes as follows: $\chr$ samples a random function $f_k \in_R F$ and a random element $x$ in the domain. $(x, y:=f_k(x))$ is sent to an adversary $\adv$, who tries to find $k'$ such that $f_{k'} (x) = y$. The $\prf$ to $\kow$ reduction is straight-line and value-dominating. Extendibility is trivial. Therefore it is $\qpoly$-respectful. This is also the case in the $\kow$ to $\eucma$-security of $\wots$ reduction. In addition $\beta$ is 1 for both reductions, which means that the effectiveness (i.e., tightness in terms of success probability) in the classical analysis carries over unchanged to the quantum setting. 

\begin{proposition}
$(\gprf, \secpar{\qpoly} = \negl(n))$ implies $(\gforot, \secpar{\qpoly} = \negl(n))_\wots$.
Namely, assuming a quantum-resistant $\prf$, $\wots$ is one-time \eucma-secure against quantum attackers $\qpoly$.
\label{prop:lwots}
\end{proposition}

\mypar{$\xmss$-tree} The $\xmss$-tree modifies Merkle's hash-tree construction with an {\sc xor}-technique. Loosely speaking, each level of the tree is associated with two random strings, which mask the two children nodes before we apply the hash function to produce an authentication of a node. This tweak allows one to use a second-preimage resistant (\spr) hash function, instead of collision-resistant hash functions or $\uowhf$. Theoretically universal one-way implies second-preimage resistance. But in practice people typically test second-preimage resistance when a hash function is designed. Despite this change, the security proof is not much different. Reductions are given that convert a forger either to a forger for $\wots$ or to an adversary that breaks $\spr$-hash functions. They are straight-line, value-dominating and $(1,\qpoly)$-extendible. By Corollary~\ref{cor:sltoresp}, we have 

\begin{proposition}
$(\gspr, \secpar{\qpoly} = \negl(n))$ and $(\gforot, \secpar{\qpoly} = \negl(n))_{\wots}$ imply\\
 $(\gfor,\secpar{\qpoly} = \negl(n))_{\xmss}$.
Namely, assuming quantum-resistant $\prf$ and $\spr$ hash functions, $\xmss$ signature is \eucma-secure against quantum attackers $\qpoly$.
\label{prop:lxmss}
\end{proposition}

As mentioned above, $\uowhf$ are by definition second-preimage resistant. As a result, quantum-resistant $\spr$ hash functions can be constructed from  quantum-resistant $\owfs$ as well. Thus, we obtain that the $\xmss$ signature scheme is \eucma-secure against efficient quantum attackers $\qpoly$, assuming quantum-resistant \owfs. 


\subsection{Full-Domain Hash in Quantum Random-Oracle Model}
\label{ssec:appgu-fdh}
Full domain hash (FDH) is a generic approach of constructing signature schemes based on trapdoor permutations (TDPs) in the RO model~\cite{BR93}. The classical proof cleverly ``programs'' the random-oracle, so that a challenge of inverting a TDP gets embedded as one output value of the random-oracle. However when we consider FDH in the quantum random-oracle  (QRO) model, in which one can query the random-oracle in superposition, we lose the ``programable'' ability in the proof. Zhandry~\cite{Zha12ibe} resolved this issue by some quantum ``programing'' strategy, which built upon lower bounds on quantum query complexity. This is summarized as follows. 

\begin{theorem}[{\cite[Theorem 5.3]{Zha12ibe}}]
Let $F$ be a quantum-resistant trapdoor permutation. If we model $H$ as a quantum random-oracle, then $\Pi$ is quantum $\eucma$-secure.
\label{thm:fdh-qro}
\end{theorem}

We note that Zhandry's proof fits our framework for lifting game-updating reductions. Namely,  let $\gtdp$ the inversion game for a $\tdp$. We can construct an interpreter $\qintp$ for any adversary in the forgery game $\gforqro$, and show that the classical reduction $(\gtdp,\trr,\gforro)$ is translatable. Applying Theorem~\ref{thm:ql-gur} proves the theorem here. We describe a proof in Appendix~\ref{apx:fdh} for completeness. This illustrates how to apply our framework and get (in our opinion) more modular security analysis.

\subsection{Quantum Security of Classical Cryptographic Protocols}
\label{ssec:appgu-sim}

So far, we have been focusing on basic cryptographic primitives such as $\uowhf$ and signatures. However, our framework is not limited to these scenarios, and actually can be applied to analyzing more complicated cryptographic protocols as well.  
Specifically an abstraction called simple-hybrid arguments, which characterize a family of classical proofs for two-party secure computation protocols in the computational setting that go through against quantum adversaries~\cite{HSS11}, can be derived easily in our framework. We defer the details in Appendix~\ref{apx:sim}.

\section{Discussions}
\label{sec:disc}

We have proposed a general framework to study which security reductions are quantum-friendly. The lifting theorems we developed can be used to analyze security against computationally bounded quantum adversaries for post-quantum cryptography. As an application, we have shown the quantum security of a generic hash-tree based signature scheme and an efficient variant (which is a promising candidate for post-quantum signature schemes to be implemented in practice). 

However, this note concerns mostly the feasibility of lifting classical security proofs to the quantum setting, and there are many important aspects missing and many interesting directions to be investigated. 
	For example, we did not consider much about the ``quality'' of the resulting proofs for quantum adversaries. Say, can we preserve the tightness of the classical reduction when we lift it? Tightness of security reduction is of great practical impact. Not only it affects how to set the parameters in implementations, it may render security meaningless in some cases~\cite{CMS12}. Interestingly, there are also examples where we get tighter reduction in the quantum setting, as demonstrated in the quantum Goldreich-Levin theorem~\cite{AC02}. This is also a nice example of game-updating reductions beyond the QRO model. Along the same line, another game-updating reduction that is fundamental in cryptography arises from constructing a pseudorandom permutation (PRP) from a pseudorandom function (PRF). It is not clear if the classical construction remains valid if the game defining PRP allows superposition queries to distinguish it from a truly random permutation.	
	
	There are many concrete questions left for quantum-secure signature schemes as well. We showed a quantum \eucma-secure signature scheme based on quantum-resistant $\owfs$. Can we make it strongly-unforgeable?  The \xmss~scheme is also known to be \emph{forward}-secure. Is it still true against quantum adversaries? {We believe both answers are positive, by similar analysis from this note}. Moreover, there are generic transformations that augments a signature scheme with stronger security guarantees (e.g., from $\eucma$-secure to $\sucma$-secure). Do they hold in the quantum setting? We also note that the applications we have shown in the game-updating case are not very exciting in the sense that designing an interpreter appears no easier than coming up with a quantum reduction directly. It is helpful to further explore along this line to find more interesting applications. 
	
	Finally, we remark that quantum attacks could reduce the security level of a system, using for example Grover's quantum search algorithm. Although not covered in this note, this issue needs to be addressed with care as well.

\section*{Acknowledgments}
\label{sec:ack}
The author is grateful to Michele Mosca for encouraging him to write up this note. F.S. would like to thank the anonymous reviewers for valuable comments and John Schank for joyful discussions on lattice-based signature schemes.  F.S. acknowledges support
from Ontario Research Fund, Industry Canada and CryptoWorks21.

\ifnum\pqc=1
\bibliographystyle{splncs03}
\else
\bibliographystyle{alpha}
\fi
\bibliography{qfred}

\appendix

\section{(Alternative) Proof of Theorem~\ref{thm:fdh-qro}: FDH in QRO}
\label{apx:fdh}

We first review a technical tool in~\cite{Zha12ibe} called \emph{semi-constant} distribution. Loosely speaking, it allows us to ``program'' a function, which still looks like a random function even to a quantum observer. 

\begin{definition}[{Semi-Constant Distribution~\cite[Definition 4.1]{Zha12ibe}}]
Let $X$ and $Y$ be sets and denote $\calH_{X,Y}$ the set of functions from $X$ to $Y$. The semi-constant distribution $SC_\lambda$ is defined as the distribution over $\mathcal{H}_{X,Y}$ resulting from the following process: 
\begin{itemize}
	\item Pick a random element $y$ from $Y$.
	\item For each $x\in X$, set $H(x) = y$ wth probability $\lambda$. Otherwise set $H(x)$ to be a random element in $Y$.
\end{itemize}
\end{definition}

\begin{theorem}[{~\cite[Corollary 4.3]{Zha12ibe}}]
The distribution of the output of a quantum algorithm making $q_H$ queries to an oracle drawn from $SC_\lambda$ is at most a distance $\frac{8}{3}q_H^4\lambda^2$ away from the case when the oracle is drawn uniformly from $\calH_{X,Y}$.
\label{thm:scd}
\end{theorem}

We are now ready to give a proof for Theorem~\ref{thm:fdh-qro} using our framework for game-updating reductions. 

\begin{proof} Classically there is $\red=(\gtdp, \trr, \gforro)$ that inverts the TDP with a forger for the FDH-Sign scheme. We construct an interpreter $\qintp$ as follows (Fig.~\ref{cons:intp-fdh}), and show that $\red$ is $(\qclsa,\qclsb)$-translatable with $\qclsa = \qclsb = \qpoly$. 

Clearly, $\qadv'$ is a well-defined (quantum) adversary for the original forgery game $\gforro$ (i.e., the hash queries are classical). If $\qadv$ outputs a valid forgery $(m^*,\sigma^*)$ such that $\hat H(m^*) = f_{pk}(\sigma^*)$ and $\calO_2(m^*) =1$, we know that $\hat H(m^*) = b = H(a)$ and hence $(a,\sigma^*)$ forms a valid forgery in the classical forgery game. Note that the view of $\qadv$ in $\qadv'$ differs from a true interaction with a challenger in game $\gforqro$ in two places: a truly random oracle is replaced by $\hat H$ drawn from $SC_\lambda$ and the signing query fails with probability $\lambda$. By picking $\lambda$ a proper inverse polynomial in $q_H$ and $q_S$, we can obtain from Theorem~\ref{thm:scd} that $\gval{\gforro}(\qintp(\qadv)) \geq \gval{\gforqro}^2(\qadv)/{p(n)}$ for some polynomial $p(\cdot)$. Thus $(\gforro,\qintp,\gforqro)$ forms a $(\beta',\qpoly)$-effective reduction for a suitable $\beta'$. Since the two random oracles $(\calO_1, \calO_2)$ can be simulated efficiently by $k$-wise indecent functions (C.f.~\cite[Theorem 6.1]{Zha12ibe}), $\red$ is clearly $\beta$-$(\qpoly,\qintp(\qpoly))$-respectful with $\beta=1$. Therefore we obtain that $\red$ is $(\qpoly,\qpoly)$-translatable, which by Theorem~\ref{thm:ql-gur} can be lifted to a reduction $(\gtdp(\qpoly), \qtrr, \gforqro(\qpoly))$. This shows that the FDH-Signature scheme is quantum \eucma-secure, assuming quantum-resistant trapdoor permutations. 
\ftodo{define games: quantumRO sign.}

\begin{prot}{Interpreter}{$\qintp$}
{\sf Input}: Adversary $\qadv$ for a quantum \eucma-game.  Let $q_S$ and $a_H$ be upper bounds on the number of signing queries and hash queries of $\qadv$.

{\sf Output}: An adversary $\qadv':=\qintp(\qadv)$ that operates as follows: 
\begin{enumerate}
	\item Receive $pk$ from a challenger, which indexes a permeation $f_{pk}$. 
	\item Pick an arbitrary message $a$. Query $H(\cdot)$ and get $b:=H(a)$. 
	\item Emulate (internally) a quantum \eucma-game with $\qadv$. 
	\begin{itemize}
		\item Use $b$ to create an oracle $\hat H$ from a semi-constant distribution $SC_\lambda$ which handles (quantum) hash queries from $\qadv$.  Specifically, let $\calO_2$ be a random oracle outputting 1 with probability $\lambda$ and $\calO_1$ be a random oracle mapping a message to an input of $f_{pk}$. Let $\hat H (x) = b$ if $\calO_2(x) = 1$ and $\hat H(x) = f_{pk}(\calO_1(x))$ otherwise.   
		\item On signing query $m_i$, if $\calO_2(m_i) =1$ abort. Otherwise respond with $\sigma_i: = \calO_1(m_i)$. 
	\end{itemize} 
	\item On output $(m^*, \sigma^*)$ from $\qadv$, if $\calO_2(m^*) = 1$ output $(a, \sigma^*)$.	
\end{enumerate}
\caption{Construction of the Interpreter.}
\label{cons:intp-fdh}
\end{prot}

\end{proof}

\section{Details on Sect.~\ref{ssec:appgu-sim}}
\label{apx:sim}

Security definitions in this setting usually follows the \emph{simulation paradigm}. In particular, there is not a simple game capturing them\footnote{In some sense, the security definitions we discussed earlier that are specified by games are \emph{falsifiable}, which does not seem to be so here.}. Roughly speaking, we require the existence of an imaginary entity (called the simulator) with certain properties for any possible adversary. The main ingredient of a security proof is often a hybrid argument, in which a sequence of imaginary experiments (a.k.a. \emph{hybrids}) are defined in terms of an adversary and the simulator. The goal is to show each adjacent pair of hybrids is indistinguishable. Whenever this is done by a reduction of breaking a computational assumption, we can define a distinguishing game (as our internal game) and study if the reduction can be lifted using our framework.  

Consider zero-knowledge proof protocols as a concrete example. Zero-knowledge property requires that for any dishonest verifier $V^*$, there is a simulator $\calS$, such that the output of $\calS$ is indistinguishable from the view of $V^*$ in real protocol with honest prover. At this moment, it looks quite alien to our framework. However, once we start the security proof, it naturally fits our framework. Basically, if we fix a dishonest $V^*$, and a specific construction of a simulator, showing that the simulator works can be thought of as a distinguishing game. 

\begin{ncprot}{ZK Distinguishing Game}{$\zkdistg$} 
{\sf Two parties}: Challenger $\chr$ and distiguisher $\disr$. 
\begin{itemize}
	\item $\chr$ flips a random coin $b\in_R\{0,1\}$. If $b=0$ simulates an execution of the ZK protocol and sends $\disr$ the view of $V^*$. If $b=1$, run the simulator $\calS$ and sends $\disr$ the output of $\calS$.
	\item $\disr$ receives the message from $\chr$, generate one bit $b'$ and send it to $\chr$.
	\item $\chr$ outputs $\succ$ if $b=b'$ and $\fail$ otherwise. 
\end{itemize}
\end{ncprot}

The security proof will then proceed in the familiar fashion. Namely a reduction $(\extg, \trr, \intg: = \zkdistg)$ is constructed for some computational assumption captured by $\extg$. We can then ask if we can ``lift'' the reduction to the quantum setting. One subtlety, however, is that the distinguishing game is specific to $V^*$ and $\calS$. Because of issues like rewinding, we have to update the games. The challenge then lies in constructing a simulator $\hat \calS$ for any dishonest quantum verifier $\hat V^*$,  which gives the updated distinguishing game $\qzkdistg$ in the presence of quantum verifiers.  

Sometimes we end up in the simpler game-preserving case. 
A concrete example is an abstraction proposed in~\cite{HSS11}, called \emph{simple-hybrid arguments} (SHA). 

\mypar{Simple hybrid arguments} SHA formalizes a family of classical proofs that can go through against quantum adversaries in the computational UC model. The essence is a simple observation: if two adjacent hybrids only differs by a small change such as chaining the plaintext of an encryption, then quantum security immediately follows as long as computational assumptions are made  quantum-resistant. Using our framework, each adjacent pair of hybrid induce a distinguishing game $\intg$ that can be defined similarly to $\zkdistg$, and a classical reduction $\red:=(\extg, \trr, \intg)$ is already at hand for some computational assumption defined by $\extg$. The conditions in SHA, e.g., changing only the plaintext, ensure that $\red$ satisfy the definition of $(\qclsa,\qclsb)$-respectful reductions with $\qclsa = \qclsb = \qpoly$. As a result, these reductions can be lifted by Theorem~\ref{thm:ql-gpr}.   

\end{document}